\documentclass{ifacconf}

\usepackage{graphicx}      
\usepackage{natbib}        
\usepackage{amsmath,amssymb,amsfonts}
\usepackage{xcolor}
\usepackage{bm}
\usepackage{enumitem}
\usepackage{ifthen}

\newboolean{MTNS}
\setboolean{MTNS}{false}

\newtheorem{them}{\it Theorem}
\newtheorem{assume}{\it Assumption}
\newtheorem{problem}{\it Problem}
\newtheorem{remark}{\it Remark}
\newtheorem{define}{\it Definition}

\newtheorem{propt}{\it Proposition}

\begin{document}
\begin{frontmatter}

\title{On the Solvability of Quasi-Regulator Equations in Non-smooth Output Regulation} 


\author[First]{Zirui Niu} 
\author[Second]{Daniele Astolfi} 
\author[First]{Giordano Scarciotti} 

\address[First]{EEE, Imperial College London, Exhibition Rd, South Kensington, London, SW7 2AZ, United Kingdom (e-mail: [zn120,gs3610]@ic.ac.uk)}
\address[Second]{Universit\'e Claude Bernard Lyon 1, CNRS, LAGEPP UMR5007, 43 boulevard du 11 novembre 1918, F-69100, Villeurbanne, France (e-mail: 
daniele.astolfi@univ-lyon1.fr)}

\begin{abstract}                
Motivated by the prevalence of non-smooth, possibly non-periodic signals in real-world applications, the output regulation of linear systems subject to non-smooth non-periodic exogenous signals has emerged as a challenging problem. A fundamental prerequisite for solving this problem is the existence of solutions to the so-called ``quasi-regulator equations''. In this paper, we investigate the solvability of these equations. To this end, we reformulate the quasi-regulator equations as differential-algebraic equations and highlight the critical role played by the system's relative degree. We finally propose a ``non-smooth non-resonance condition'' that, under specific relative degree requirements, provides a necessary and sufficient characterization of the solvability of the quasi-regulator equations.
\end{abstract}


\end{frontmatter}

\section{Introduction}\label{sec:intro}
Output regulation, a classical problem in control theory, entails designing a regulator to ensure that the output of a dynamic system asymptotically tracks the reference trajectories while rejecting disturbances. In this problem, both the references and disturbances, called the exogenous signals, are assumed to be generated by an autonomous dynamic system, commonly referred to as the exosystem. Solving an output regulation problem typically hinges on the solutions of the so-called regulator equations, and therefore, characterizing the conditions that determine their solvability is a significant topic. For example, in the classical linear time-invariant (LTI) case established by~\cite{ref:davison1976robust, ref:francis1976internal, ref:francis1977linear, ref:wonham1974linear}, such solvability of the regulator equations boils down to the well-known non-resonance condition, which requires the transmission zeros of the controlled linear system to be disjoint from the eigenvalues of the exosystem. The regulator equations can be extended to cases with more general classes of exogenous signals, see, \textit{e.g.},~\cite{ref:marconi2025robust}, and also play an important role in the output regulation of nonlinear systems, see, \textit{e.g.},~\cite{ref:byrnes2003limit,ref:petit2018necessary}.

While the majority of existing output regulation literature assumes differentiability of exogenous signals or considers piecewise absolutely continuous signals with periodic jumps (see, \textit{e.g.},~\cite{ref:huang2004nonlinear,ref:serrani2001semi,ref:chen2005robust,ref:marconi2007output,ref:zhang2009adaptive,ref:marconi2013internal, ref:carnevale2015hybrid, ref:carnevale2017robust}), non-smooth signals, possibly with non-periodic dynamics, such as sawtooth or pulse width modulation (PWM) signals, are commonly encountered in real-life applications, see, \textit{e.g.},~\cite{ref:sira1996dynamical,ref:cortes2008discontinuous, ref:chen2013tracking, ref:akhmet2010principles}. Motivated by this fact, recent studies presented by~\cite{ref:niu2024full,ref:niu2024adaptive,ref:Niu2025arXiv} focused on the output regulation problem for LTI systems subject to exogenous signals with \textit{non-smooth} and \textit{non-periodic} dynamics. For brevity, such a problem is referred to herein as the non-smooth non-periodic output regulation (NNOR) problem. While the NNOR problem has been addressed in different scenarios, such as the full-information or error-feedback case, all proposed solutions require solving a set of equations analogous to, yet distinct from, the standard regulator equations, herein called the \textit{quasi-regulator equations}\footnote{Due to the generality of the NNOR problem setting, regulator equations that characterize the steady-state error-zeroing dynamics may not be well-defined (as the steady state may be ill-defined), nor are they necessary for solving this problem. Instead,~\cite{ref:Niu2025arXiv} demonstrates that the solutions of the quasi-regulator equations are sufficient for solving the problem.}. Determining the solvability conditions for such equations is crucial but remains an open problem. 

This paper investigates the solvability conditions for the quasi-regulator equations arising in the NNOR problem, which involve integral terms that present analytical challenges (to be clarified later). To address this challenge, we first reformulate the quasi-regulator equations as a group of differential-algebraic equations (DAEs). Subsequently, we examine the solvability of these DAEs by analysing two different configurations separately: systems without a feedforward path and systems with a feedforward path. Note that the latter case proves more intricate, requiring a detailed analysis of the system's relative degree and the smoothness of the exogenous signal. Finally, we establish a ``non-smooth non-resonance condition'' and demonstrate that, under certain relative degree requirements, the non-smooth non-resonance condition is a necessary and sufficient condition for the quasi-regulator equations to be solvable. It is worth noting that this solvability problem was preliminarily introduced in~\cite{ref:niu2024full}. However, that study presented only a sufficient condition for a simplified case and lacks rigorous technical proofs. Furthermore, the results therein are restricted to a limited class of explicit generators that satisfy semigroup properties.

The rest of the paper is organized as follows. Section~\ref{sec:Prelim} recalls the formulation of the NNOR problem and its corresponding quasi-regulator equations. Section~\ref{sec:solvability} studies the solvability of the quasi-regulator equations by transforming them into DAEs and considering the feedforward case and non-feedforward case separately. Section~\ref{sec:concl} concludes the paper.

\textbf{Notation.} $\mathbb{R}_{\geq 0}$ denotes the set of non-negative real numbers, $\mathbb{R}_{>0}$ indicates $\mathbb{R}_{\geq 0} \backslash\{0\}$, 
$\mathbb{C}_{<0}$ stands for the set of complex numbers with a strictly negative real part and $\mathbb{C}_{\geq 0}$ denotes $\mathbb{C} \backslash \mathbb{C}_{<0}$. The symbol $\otimes$ indicates the Kronecker product, $I_{n}$ denotes an $n \times n$ identity matrix, $\bm{0}_{m \times n}$ represents an $m \times n$ zero matrix, $\sigma(A)$ represents the spectrum of the matrix $A \in \mathbb{R}^{n \times n}$ and $\|A\|$ indicates its induced Euclidean matrix norm. The superscript $\top$ stands for the transposition operator. For $n$ matrices $M_{i} \in \mathbb{R}^{p \times q}$, $i = 1, 2, \cdots, n$, $\operatorname{col}\left(M_1, M_2, \cdots, M_n\right)$ indicates the $pn \times q$ matrix obtained by vertically stacking the matrices $M_i$. Given a function $H(\cdot)$ and a positive integer $k$, the notation $H \in \mathcal{C}^{k}$ indicates that $H$ is $k$-times differentiable with its $k$-th derivative, denoted by $H^{(k)}$, being continuous in its domain. In addition, given an initial value $t_0$, the functional $\mathcal{I}^{[k]}_{t_0} :  H \mapsto \mathcal{I}^{[k]}_{t_0}[H]$ denotes $k$-times repeated integral of the function $H$, namely
$$
\mathcal{I}^{[k]}_{t_0}[H(t)] =\int_{t_0}^t \int_{t_0}^{\sigma_1} \cdots \int_{t_0}^{\sigma_{k-1}} H\left(\sigma_k\right) \mathrm{d} \sigma_{k} \cdots \mathrm{d} \sigma_2 \mathrm{~d} \sigma_1.
$$

\section{Preliminaries and Quasi-Regulator Equations}\label{sec:Prelim}
In this section, we introduce the formulations of an NNOR problem along with its quasi-regulator equations. To this end, consider a single-input single-output (SISO), LTI system of the form
\begin{equation}\label{equ:system}
    \begin{aligned}
    \dot{x}(t) &= Ax(t) + Bu(t) + P\omega(t),\\
    e(t) &= Cx(t) + Du(t) + Q\omega(t),
    \end{aligned}
\end{equation}
with $A \in \mathbb{R}^{n \times n}$, $B \in \mathbb{R}^{n \times 1}$, $P \in \mathbb{R}^{n \times \nu}$, $C \in \mathbb{R}^{1 \times n}$, $D \in \mathbb{R}$, $Q \in \mathbb{R}^{1 \times \nu}$. $x(t) \in \mathbb{R}^n$ the state, $u(t) \in \mathbb{R}$ the control input, $e(t) \in \mathbb{R}$ the regulation error, and $\omega(t) = [r(t),\; d(t)^{\top}]^{\top} \in \mathbb{R}^{v}$ the exogenous signal representing the disturbances $d(t) \in \mathbb{R}^{\nu-1}$ and/or reference signal $r(t) \in \mathbb{R}$. 

As mentioned in the introduction, a recent development in the output regulation field focuses on a class of general exogenous signals that have \textit{non-smooth} and \textit{non-periodic} dynamics. This class of signals is assumed to be generated by an explicit generator of the form \begin{equation}\label{equ:explicitGen}
\omega(t) = \Lambda\left(t, t_0\right) \omega_0.
\end{equation}
with $\Lambda\left(t, t_0\right) \in \mathbb{R}^{\nu \times \nu}$. This generator, as discussed in~\cite[Section 5.1]{ref:niu2024full,ref:Niu2025arXiv,ref:scarciotti2017nonlinear}, can represent signals produced by various autonomous systems, including linear time-varying and hybrid models. Note that similarly to~\cite{ref:Niu2025arXiv} but differently from~\cite{ref:niu2024full}, in this paper the function $\Lambda$ does not necessarily satisfy the standard semigroup property, \textit{i.e.}, $\Lambda\left(t_2, t_1\right) = \Lambda\left(t_2, t_1\right)\Lambda\left(t_1, t_0\right)$. In the NNOR problem setting, the following assumption is introduced.

\begin{assume}\label{asmp:exoProp}
The matrix-valued function $\Lambda$ is piecewise continuous, non-singular, and finite-time bounded with $\Lambda(\tau, t_0)\Lambda(t, t_0)^{-1}$ uniformly bounded for all $t \ge \tau \ge t_0$.
\end{assume}

This assumption is not restrictive as discussed in detail in~\cite{ref:Niu2025arXiv}.

We are now ready to introduce the NNOR problem. Under Assumption~\ref{asmp:exoProp}, the NNOR problem consists in designing a regulator such that:
\begin{enumerate}[leftmargin=3em, labelsep=0.5em, align=left, labelwidth=2.5em]
    \item[\textbf{($\mathbf{S_F}$)}] The origin of the closed-loop system obtained by interconnecting system~(\ref{equ:system}) and the regulator with $\omega(t) \equiv 0$ is exponentially stable.
    \item[\textbf{($\mathbf{R_F}$)}] The closed-loop system obtained by interconnecting system~(\ref{equ:system}), the exosystem~(\ref{equ:explicitGen}), and the regulator satisfies $\lim_{t \rightarrow +\infty} e(t) = 0$
    for all $(x(t_0), \omega(t_0)) \in \mathbb{R}^{n + \nu}$.
\end{enumerate}

Solving such an NNOR problem, as shown in~\cite{ref:Niu2025arXiv}, typically requires finding bounded piecewise continuous $\Pi(t) \in \mathbb{R}^{n \times \nu}$ and $\Delta(t) \in \mathbb{R}^{1 \times \nu}$ that solve
\begin{equation}\label{equ:reguEquation}
    \begin{aligned}
    \Pi_x(t) \!&=\! \Big(\! e^{A (t-t_0)} \Pi_x(t_0)\Lambda(t_0, t_0) \\
    &\qquad \quad+\! \int_{t_0}^t \!e^{A(t-\tau)} P_{\Delta}(\tau) \Lambda(\tau, t_0) d \tau  \Big) \Lambda(t, t_0)^{-1}\!, \\
    \bm{0}_{1 \times \nu} \!&= C\Pi_x(t) + D\Delta(t) + Q,
    \end{aligned}
\end{equation}
for all $t \geq t_0$, with $P_{\Delta}(t) = P+B\Delta(t)$. As mentioned in Section~\ref{sec:intro}, differently from classical regulator equations, the equations~\eqref{equ:reguEquation} serve as a sufficient, but not necessary, condition for solving the NNOR problem due to the generality of the problem setting. Nevertheless, their solutions $\Pi$ and $\Delta$ are pivotal as they characterize the expected trajectories of state $x^* = \Pi\omega$ and input $u^*= \Delta \omega$ that render the regulation error $e$ identically zero, see~\cite{ref:Niu2025arXiv} for more details. In this context, the equations~\eqref{equ:reguEquation} are referred to as the \textit{quasi-regulator equations} and, accordingly, the solvability problem to be addressed in this paper is formulated as follows.
\begin{problem}\label{def:solvable}(Solvability of Quasi-Regulator Equations). 
Suppose Assumption~\ref{asmp:exoProp} holds. Determine the conditions under which the quasi-regulator equations~\eqref{equ:reguEquation} are solvable, \textit{i.e.}, there exist bounded piecewise continuous $\Pi(t) \in \mathbb{R}^{n \times \nu}$ and $\Delta(t) \in \mathbb{R}^{1 \times \nu}$ such that~\eqref{equ:reguEquation} holds for all $t \geq t_0$.
    \hfill $\blacksquare$
\end{problem}

\section{Solvability of the Quasi-Regulator Equations}\label{sec:solvability}
In this section, we investigate the solvability of the quasi-regulator equations~\eqref{equ:reguEquation}. Since these equations are expressed in an integral form, characterizing their solvability poses inherent analytical challenges. To overcome this issue, we first recall that it is possible to transform~\eqref{equ:reguEquation} into a more tractable representation.

\begin{lem}\label{lem:equivDAE}
Suppose Assumption~\ref{asmp:exoProp} holds. The quasi-regulator equations~(\ref{equ:reguEquation}) are solvable if and only if the DAEs
\begin{subequations}\label{equ:solutionDAE}
    \begin{align}
    \dot{\Psi}_x(t) &= A\Psi_x(t) + (B\Delta(t) + P)\Lambda(t, t_0), \label{equ:solutionDAE_a}\\
    \bm{0}_{1 \times \nu} &= C \Psi_x(t) + (D\Delta(t) + Q)\Lambda(t, t_0), \label{equ:solutionDAE_b}
    \end{align}
\end{subequations}
are solvable, \textit{i.e.}, there exist a finite-time bounded absolutely continuous $\Psi_x(t) \in \mathbb{R}^{n \times \nu}$ and a bounded piecewise continuous $\Delta(t) \in \mathbb{R}^{1 \times \nu}$ satisfying~\eqref{equ:solutionDAE} for almost all $t \geq t_0$.
\end{lem} 
\begin{pf}
    \cite[Proof of Corollary~1]{ref:Niu2025arXiv} shows that~\eqref{equ:reguEquation} and~\eqref{equ:solutionDAE} are equivalent by $\Psi_x(t) = \Pi_x(t) \Lambda(t, t_0)$ with $\Lambda$ invertible for all times.
    \hfill $\blacksquare$
\end{pf}

By the equivalence established in Lemma~\ref{lem:equivDAE}, the solvability of the quasi-regulator equations~\eqref{equ:reguEquation} is addressed in the remainder of the paper through the analysis of the DAEs in~\eqref{equ:solutionDAE}. To this end, we explore the role of the relative degree in the solvability of the problem and propose a non-smooth version of the classical non-resonance condition.

Given the generality of $\Lambda$ considered so far, the presence of the feedforward term $D$ (\textit{i.e.}, zero relative degree) is significant for the solvability of~\eqref{equ:solutionDAE}. This is formalised in the next result

\begin{them}\label{thm:zeroError}
Consider the DAEs in~\eqref{equ:solutionDAE}. Suppose $D = 0$ and Assumption~\ref{asmp:exoProp} holds. Then~\eqref{equ:solutionDAE}, or equivalently~\eqref{equ:reguEquation}, is solvable only if the function $Q\Lambda(t, t_0)$ is locally Lipschitz continuous for all $t \geq t_0$.
\end{them}
\ifthenelse{\boolean{MTNS}}{See~\cite{ref:niu2026solvabilityArXiv} for a proof of Theorem~\ref{thm:zeroError}.}{
\begin{pf}
Suppose the DAEs in~\eqref{equ:solutionDAE} are solvable. Since $\Lambda$ is finite-time bounded by assumption,~\eqref{equ:solutionDAE_a} implies the finite-time boundedness of $\dot{\Psi}$, indicating that $\Psi_x$ is locally Lipschitz continuous, see~\cite[Chapter 3]{ref:khalil2002nonlinear}. This means that for any time constants $t_1$ and $t_2$ greater than $t_0$, there exists a Lipschitz constant $L_{\Psi} \in \mathbb{R}_{\geq 0}$ such that $\|\Psi(t_1)-\Psi(t_2)\| \leq L_{\Psi}\|t_1-t_2\|$. Then the satisfaction of~\eqref{equ:solutionDAE_b} yields that $C\Psi(t_1) + Q\Lambda(t_1, t_0) = \bm{0}_{1 \times \nu}$ and $C\Psi(t_2) + Q\Lambda(t_2, t_0) = \bm{0}_{1 \times \nu}$ for all $t \geq t_0$. Subtracting these two equations, we get $C(\Psi(t_1) - \Psi(t_2)) + Q(\Lambda(t_1, t_0) - \Lambda(t_2, t_0)) = \bm{0}_{1 \times \nu}$, which yields
$\|Q(\Lambda(t_1, t_0) - \Lambda(t_2, t_0))\| = \|C(\Psi(t_1) - \Psi(t_2))\| \leq \|C\|\|\Psi(t_1) - \Psi(t_2)\| \leq \|C\|L_{\Psi}\|t_1-t_2\|$. Therefore, $Q\Lambda(t, t_0)$ is locally Lipschitz continuous for all $t \geq t_0$.
\hfill $\blacksquare$
\end{pf}
}
Theorem~\ref{thm:zeroError}, together with Lemma~\ref{lem:equivDAE}, yields that the quasi-regulator equations~\eqref{equ:reguEquation} cannot be solved by a bounded $\Delta$ without the feedforward term $D$ when $Q\Lambda$ is discontinuous. Recall that the expected error-zeroing state and input dynamics are characterized by $x^* = \Pi\omega$ and $u^* = \Delta\omega$. Then the result of Theorem~\ref{thm:zeroError} implies that for exogenous signals where $Q\omega$ presents discontinuities such finite-time bounded error-zeroing dynamics $x^*$ and $u^*$ cannot exist without a feedforward term\footnote{Note that if $u$ is allowed to be an impulsive control, then the problem with discontinuous signals may be solved without a feedforward term. This is left as a future research direction.}. Note that this makes sense: a discontinuous signal cannot be tracked by $Cx$, which is a continuous signal, without the help of a feedforward term $Du$ that can cancel the jump. In the rest of this section, we mainly focus on the case $D = 0$, as the case $D\neq 0$ is much simpler. We summarise the discussion above in the following assumption that guarantees the existence of a solution to the regulation problem when $D = 0$.


\begin{assume}\label{asmp:pwDiff}
The matrix-valued function $Q\Lambda(t, t_0)$ is piecewise differentiable\footnote{By definition, this means that the function $Q\Lambda(t, t_0)$ is continuous and for any $t_b > t_a \geq t_0$, there exists a finite subdivision $t_a = \hat{t}_0 < \hat{t}_1 < \cdots <\hat{t}_{k-1} < \hat{t}_{k} = t_b$ of $[t_a,\;t_b]$ such that $Q\Lambda(t, t_0)$ is continuously differentiable in each subinterval $[t_{i-1},\; t_{i}]$ for any $i = 1, 2, \cdots, k$. Note that the derivative at $t_{i-1}$ is understood as the right derivative and the derivative at $t_{i}$ is understood as the left derivative, see~\cite[Definition 3.1]{ref:do1992riemannian}.} and $\frac{d}{dt}\bigl(Q\Lambda(t, t_0)\bigr)\Lambda(t, t_0)^{-1}$ is bounded for all $t \in \mathcal{T}$, where $\mathcal{T}$ denotes any time interval in which $Q\Lambda(t, t_0)$ is continuously differentiable.
\end{assume}

Assumption~\ref{asmp:pwDiff} is a stronger assumption compared with the local Lipschitz continuity requirement on $Q\Lambda$ (or equivalently, $Q\omega$), which is necessary by Theorem~\ref{thm:zeroError}. On the one hand, the piecewise differentiability assumption implies that $Q\Lambda$ is continuous, has a finite number of non-differentiable points in any finite interval, and is semi-differentiable at each non-differentiable point, \textit{i.e.}, both left and right derivatives exist, although different. This differentiability requirement does not practically restrict the applicability of the result when considering exogenous signals such as (possibly time-varying) triangular waveforms. On the other hand, the boundedness of $\frac{d}{dt}\bigl(Q\Lambda(t, t_0)\bigr)\Lambda(t, t_0)^{-1}$ can be seen as a direct extension of the continuous-time output regulation problem to the non-smooth case. More specifically, if the exosystem~\eqref{equ:explicitGen} can be equivalently expressed as
\begin{equation}\label{equ:timeVaryingGen}
\dot{\omega}(t)= \tilde{S}(t) \omega(t), \qquad \omega(t_0) = \omega_0,
\end{equation}
with $\tilde{S}(t) \in \mathbb{R}^{\nu \times \nu}$, then $\Lambda$ is the corresponding state-transition matrix and the boundedness condition on $\frac{d}{dt}\bigl(Q\Lambda(t, t_0)\bigr)\Lambda(t, t_0)^{-1} = Q\dot{\Lambda}(t, t_0)\Lambda(t, t_0)^{-1} = Q\tilde{S}(t)$ naturally holds as $\tilde{S}(t)$ is classically assumed to be bounded for all times in the time-varying setting, see, \textit{e.g.},~\cite{ref:zhang2009adaptive}. It should be pointed out that Assumption~\ref{asmp:pwDiff} on $Q\Lambda$ does not imply that $\Lambda$, or equivalently $\omega$, cannot be discontinuous. In fact, the solvability can still be satisfied when the jumps of $\omega$ occur in the kernel of the matrix $Q$ for all times. In other words, we do not exclude discontinuous signals from $P\omega$.

We also note that when $D = 0$, the solvability of the problem is closely related to the relative degree\footnote{The smallest positive integer $r \leq n$ such that $C A^{r-1} B \neq 0$.} of system~(\ref{equ:system}). While this may seem surprising (because it is well known that the relative degree does not play a role in the solvability of the classical regulator equations), in our problem setting, the relative degree controls the ``degree of non-smoothness'' of the state $x$. In other words, by superposition, the relative degree determines the ability of the system~(\ref{equ:system}) with $\omega \equiv 0$ and a finite-time bounded input $u$ to produce a (possibly non-smooth) output that can cancel the (possibly non-smooth) output of system~(\ref{equ:system}) with $u \equiv 0$. Therefore, when $D = 0$, we study the solvability of the quasi-regulator equations~(\ref{equ:reguEquation}) with an eye to the relative degree of system~(\ref{equ:system}). For the time being, under Assumption~\ref{asmp:pwDiff}, define a group of matrix-valued functions $V_{j}(t) \in \mathbb{R}^{1 \times \nu}$ as
%
%
\begin{equation}\label{equ:functionsVj}
    V_{j}(t) := \sum_{i=1}^{j} \mathcal{I}^{[i]}_{t_0} \left[CA^{i-1}P\Lambda(t, t_0)\right] + Q\Lambda(t, t_0)
\end{equation}
for $j = 1, 2, \cdots, n$, with $V_0(t) = Q\Lambda(t, t_0)$. Moreover, define a sequence of positive integers $S_c = \{c_0, c_1, \cdots, c_n\}$ with each $c_{j}$ denoting the degree of smoothness of the corresponding matrix-valued function $V_{j}$, \textit{i.e.}, $V_{j} \in \mathcal{C}^{c_j}$. 
The idea behind the definition of the functions $V_{j}$ can be explained as follows: we cannot use differentiation as a tool to prove our claims because of the non-smoothness of $\Lambda$. Thus, we define auxiliary functions $V_{j}$ as ``integrals" that have sufficient smoothness to be differentiated as needed. As it will be clear later, the derivatives of the functions $V_{j}$ will play a role in the proof of the solvability condition. We now introduce a preliminary technical lemma.

\begin{lem}\label{lem:VjProperty}
    Consider the group of functions~(\ref{equ:functionsVj}). Let $j^{*} = c_n$. Suppose $j^{*} < n$. Then the following properties hold.
    \begin{itemize}
    \item[(i)] For all $j \in \{0, 1, \cdots, j^{*}\}$, $c_j \geq j$;
    \item[(ii)] For all $j \in \{j^{*}+1 \cdots, n\}$, $c_{j} = j^{*}$. 
\end{itemize}
\end{lem}
\ifthenelse{\boolean{MTNS}}{See~\cite{ref:niu2026solvabilityArXiv} for a proof of Lemma~\ref{lem:VjProperty}.}{
\begin{pf}
    We first prove (ii). Considering $j^{*} = c_n < n$, the definition in~(\ref{equ:functionsVj}) implies that 
    $$
    V_{n}(t) = \sum_{i=j^{*}+2}^{n} \mathcal{I}^{[i]}_{t_0} \left[CA^{i-1}P\Lambda(t, t_0)\right] + V_{j^{*}+1}(t).
    $$
    Since the term $\sum_{i=j^{*}+2}^{n} \mathcal{I}^{[i]}_{t_0}\left[CA^{i-1}P\Lambda(t, t_0)\right]$ is at least ($j^{*}+1$)-times differentiable, the degree of smoothness $c_n = j^{*}$ of function $V_n$ can only be induced by the term $V_{j^{*}+1}$. In other words, $V_n \in  \mathcal{C}^{c_n}$ with $c_n = j^{*} < n$ only if $c_{j^{*}+1} = j^{*} = c_n$, where $c_{j^{*}+1}$ is the degree of smoothness of $V_{j^{*}+1}$. Moreover, as 
    $$
    V_{j}(t)= \sum_{i=j^{*}+2}^{j} \mathcal{I}^{[i]}_{t_0} \left[CA^{i-1}P\Lambda(t, t_0)\right] + V_{j^{*}+1}(t)
    $$
    for all $j^{*}+1 < j \leq n$ with $\sum_{i=j^{*}+2}^{j} \mathcal{I}^{[i]}_{t_0}\left[CA^{i-1}P\Lambda(t, t_0)\right]$ at least ($j^{*}+1$)-times differentiable, the fact $c_{j^{*}+1} = j^{*}$ also yields that $c_{j} = c_{j^{*}+1} = j^{*}$ for all $j^{*}+1 < j \leq n$, \textit{i.e.}, the argument (ii) holds. The argument (i) can be proved in a similar manner. For any $j \in \{0, 1, \cdots, j^{*}\}$, we have 
    $$
    V_{j^{*}+1}(t) = \sum_{i=j+1}^{j^{*}+1} \mathcal{I}^{[i]}_{t_0} \left[CA^{i-1}P\Lambda(t, t_0)\right] + V_{j}(t).
    $$
    Assume, by contradiction, that $V_{j} \in C^{c_j}$ with $c_j < j \leq j^{*}$. Since the term $\sum_{i=j+1}^{j^{*}+1} \mathcal{I}^{[i]}_{t_0} \left[CA^{i-1}P\Lambda(t, t_0)\right]$ is at least $j$-times differentiable, $V_{j^{*}+1} \in C^{c_{j^{*}+1}}$ with $c_{j^{*}+1} = c_j < j^{*}$, which contradicts the equality $c_{j^{*}+1} = j^{*}$ proved before. Therefore, $c_j \geq j$ for any $j \in \{0, 1, \cdots, j^{*}\}$.
    \hfill $\blacksquare$
\end{pf}
}
Lemma~\ref{lem:VjProperty}, in short, reflects two facts: i) if the degree of smoothness, $c_n$, of the term $V_n$ satisfies $c_n < n$, this degree of smoothness coincides with the degree of smoothness of the ($c_n$+1)-th term in the sequence of functions $V_j$'s; ii) all functions $V_j$ with $j \leq c_n$ are at least $j$-times differentiable. These properties are instrumental in investigating the role of the relative degree of system~(\ref{equ:system}) in the solvability of the quasi-regulator equations~(\ref{equ:reguEquation}). With this consideration in mind, the following proposition establishes a necessary condition for the solvability of the quasi-regulator equations~(\ref{equ:reguEquation}) when $D = 0$.

\begin{propt}\label{prop:RDCondition}
Consider the quasi-regulator equations~(\ref{equ:reguEquation}). Suppose $D = 0$ and Assumptions~\ref{asmp:exoProp} and~\ref{asmp:pwDiff} hold. The quasi-regulator equations~(\ref{equ:reguEquation}) are solvable only if the relative degree $r$ of system~(\ref{equ:system}) satisfies $r \leq j^{*} + 1$, where $j^{*}=c_n$.
\end{propt}

\begin{pf}
When $D = 0$, by Lemma~\ref{lem:equivDAE}, the solvability of~(\ref{equ:reguEquation}) is equivalent to that of
\begin{equation}\label{equ:solutionDAEwithoutD}
    \begin{aligned}
    \dot{\Psi}_x(t) &= A\Psi_x(t) + (B\Delta(t) + P)\Lambda(t, t_0), \\
    \bm{0}_{1 \times \nu} &= C \Psi_x(t) + Q\Lambda(t, t_0). \\
    \end{aligned}
\end{equation}
If $j^* \geq n$, then $r \leq n$ is surely smaller than $j^*+1$. Thus, consider the case $j^{*} < n$. Assume by contradiction that $r > j^{*} + 1$. By Lemma~\ref{lem:VjProperty}(ii), $c_{j^*+1}= j^*$, which means by definition that $V_{j^{*}+1} \in \mathcal{C}^{j^{*}}$ and, by (i), $V_{j}$ is at least $j$-times differentiable for all $j \leq j^{*}$. Then, note that the second equation in~(\ref{equ:solutionDAEwithoutD}) can be written as
$\bm{0}_{1 \times \nu} = C\Psi(t) + V_0^{(0)}(t)$,
and that the equation
$\bm{0}_{1 \times \nu} = CA\Psi(t) + CB\Delta(t)\Lambda(t, t_0) + V_1^{(1)}(t)$
is its first derivative. By iterating this for $j^{*}$ times we obtain
\begin{equation}\label{equ:solutionDAEwithoutDDiff}
    \bm{0}_{1 \times \nu} = CA^{j^{*}}\Psi_x(t) + CA^{j^{*}-1}B\Delta(t)\Lambda(t, t_0) + V_{j^{*}}^{(j^{*})}(t).
\end{equation}
Note that $V_{j^*}$ is $j^*$-times differentiable, and thus~(\ref{equ:solutionDAEwithoutDDiff}) is well defined. Then, we substitute the integral solution of the first equation in~(\ref{equ:solutionDAEwithoutD}) into~(\ref{equ:solutionDAEwithoutDDiff}), and we use the fact that $CA^i B=0$ for all $i\le j^*$ because $r>j^*+1$, obtaining
\begin{equation}\label{equ:solutionDAEwithoutDDiffCont}
    \bm{0}_{1 \times \nu}\! = \!\!\int_{t_0}^{t} \!  CA^{j^{*}+1}\Psi_x(\tau) + CA^{j^{*}}\!P\Lambda(\tau, t_0) d\tau + \hat{G} + V_{j^{*}}^{(j^{*})}(t),
\end{equation}
for all $t \geq t_0$, where $\hat{G} = CA^{j^{*}}\Psi_x(t_0)$ is a constant matrix. Now considering that $V_{j^{*}+1}(t) = V_{j^{*}}(t) + \mathcal{I}^{[j^{*}+1]}_{t_0} \left[CA^{j^{*}}P\Lambda(t, t_0)\right]$, we have $V_{j^{*}+1}^{(j^{*})}(t) = V_{j^{*}}^{(j^{*})}(t) + \int_{t_0}^{t} CA^{j^{*}}P\Lambda(\tau, t_0) d\tau$. Then~(\ref{equ:solutionDAEwithoutDDiffCont}) can be rewritten as
\begin{equation*}
    -\hat{G} =\!\! \int_{t_0}^{t}\! CA^{j^{*}+1}\Psi_x(\tau) d\tau + V_{j^{*}+1}^{(j^{*})}(t).
\end{equation*}
Since $V_{j^{*}+1} \in \mathcal{C}^{j^{*}}$, $V_{j^{*}+1}^{(j^{*})} \in \mathcal{C}^{0}$. By the fact that $\int_{t_0}^{t} CA^{j^{*}}\Psi_x(\tau) d\tau$ is continuously differentiable and cannot generate functions of class $\mathcal{C}^{0}$, the last equation is not solvable. Therefore, the quasi-regulator equations~(\ref{equ:reguEquation}) are solvable only if $r \leq j^* + 1$.
\hfill $\blacksquare$
\end{pf}

\begin{remark}
Proposition~\ref{prop:RDCondition} discusses the relation between the non-smooth output dynamics of system~(\ref{equ:system}) induced by $\omega$ and the relative degree that is necessary for cancelling the resulting non-smoothness via an external finite-time bounded piecewise continuous input $u$. However, the requirement $r \leq j^{*} + 1$ may not sufficiently guarantee this cancellation. For example, if $V_{j^{*}+1}^{(j^{*})}$ is continuous with a derivative that at some time instants is not finite, then (\ref{equ:solutionDAEwithoutDDiffCont}) cannot be satisfied when $r = j^{*} + 1$ because $CA^{j^{*}-1}B = 0$ and $CA^{j^{*}}\Psi_x$ always has a finite-time bounded derivative under the control of a finite-time bounded input $u$.
\end{remark}

In fact, the case $j^{*} = c_n \geq 1$ implies that the regulation error $e$ of the linear system~(\ref{equ:system}) subject to the explicit generator~(\ref{equ:explicitGen}) with $u \equiv 0$ is continuously differentiable. This can be seen by the fact that $\dot{e}(t) = CAx(t) + \dot{V}_1(t)\omega_0$ where, by Lemma~\ref{lem:VjProperty}, $V_1$ is continuously differentiable. In the non-smooth case that we target (possibly time-varying triangular waves and square waves), the case $c_n \geq 1$ is less likely to happen. Therefore, in the remainder of this section, we mainly focus on the case $r = 1$, which, by Proposition~\ref{prop:RDCondition}, is necessary for guaranteeing the solvability of the quasi-regulator equations for any $\Lambda$ under Assumptions~\ref{asmp:exoProp} and~\ref{asmp:pwDiff}. To this end, we introduce the following assumption.
\begin{assume}\label{asmp:unitaryRD}
System~(\ref{equ:system}) has a unitary relative degree.
\end{assume}

Then, we study the solvability condition for the regulation problem under extra Assumptions~\ref{asmp:pwDiff} and~\ref{asmp:unitaryRD}. For the time being, suppose Assumption~\ref{asmp:pwDiff} holds. Then for any compact time interval $[t_a, t_b] \subset [t_0, +\infty)$, there exists a finite subdivision $t_a = \hat{t}_0 < \hat{t}_1 < \cdots <\hat{t}_{k-1} < \hat{t}_{k} = t_b$ of $[t_a, t_b]$  such that the function $Q\Lambda(t, t_0)$ is continuously differentiable in each subinterval $[\hat{t}_{j-1},\; \hat{t}_{j}]$ for any $j = 1, 2, \cdots, k$. Moreover, define a matrix-valued function $Q_{\Lambda}: \mathbb{R} \to \mathbb{R}^{1 \times \nu}$ such that $Q_{\Lambda}(t, t_0) := \frac{d}{dt}\bigl(Q\Lambda(t, t_0)\bigr)\Lambda(t, t_0)^{-1}$ for all $t \in [\hat{t}_{j-1},\; \hat{t}_{j}) \subset [t_a, t_b]$ with $j = 1, 2, \cdots, n$ and for any subset $[t_a, t_b] \subset [t_0, +\infty)$. Then $Q_{\Lambda}$ satisfies
\begin{equation*}
    Q\Lambda(t, t_0) = Q\Lambda(t_0, t_0) + \int_{t_0}^{t} Q_{\Lambda}(\tau, t_0)\Lambda(\tau, t_0) d\tau
\end{equation*}
and, by Assumption~\ref{asmp:pwDiff}, is bounded and piecewise continuous for all times. We now check the solvability by introducing a new non-resonance condition.

\begin{define}\label{def:nonResCond}(Non-smooth Non-resonance Condition).
Systems~(\ref{equ:system}) and~(\ref{equ:explicitGen}) are \textit{non-resonant} if there exists an initial condition $\Omega(t_0) \in \mathbb{R}^{(n-r)\nu \times (n-r)\nu}$ such that the matrix-valued function
\begin{equation}\label{equ:nonSmNRCOmega}
\begin{aligned}
    \Omega(t)= &\left( \left(\Lambda(t_0, t_0)\Lambda(t, t_0)^{-1}\right)\!\!^\top \otimes e^{A_{z}(t-t_0)} \right)\Omega(t_0) \\
    & \quad + \int_{t_0}^t \left(
    \Lambda(\tau, t_0)\Lambda(t, t_0)^{-1}\right)^{\top} \otimes e^{A_{z}(t - \tau)} d \tau,
\end{aligned}
\end{equation}
is bounded for all $t \geq t_0$, where $A_{z} \in \mathbb{R}^{(n-r) \times (n-r)}$ has eigenvalues identical to the transmission zeros of system~(\ref{equ:system}).
 \end{define}

Leveraging the non-resonance condition, it is possible to obtain the following result when $D = 0$. 

\begin{them}\label{thm:SolCondWithoutD}
Consider the quasi-regulator equations~(\ref{equ:reguEquation}). Suppose $D = 0$ and Assumptions~\ref{asmp:exoProp},~\ref{asmp:pwDiff}, and~\ref{asmp:unitaryRD} hold. The quasi-regulator equations~(\ref{equ:reguEquation}) are solvable for any $P$ and $Q$ if and only if systems~(\ref{equ:system}) and~(\ref{equ:explicitGen}) are non-resonant.
\end{them}



\begin{pf} 
We first prove condition (i). For simplicity, let $b = CB \neq 0$ under Assumption~\ref{asmp:unitaryRD}. Consider the DAEs in~\eqref{equ:solutionDAE}. We define a linear transformation $T := [T_1^{\top},\;T_2^{\top}]^{\top}$ with $T_2 = C$ and $T_1$ being selected such that $T$ is non-singular. By the change of variable $\Psi_x \mapsto \Theta = T\Psi_x := \operatorname{col}\left(\Theta_z, \Theta_y\right)$ with $\Theta_z \in \mathbb{R}^{(n-1) \times \nu}$ and $\Theta_y \in \mathbb{R}^{1 \times \nu}$,~(\ref{equ:solutionDAEwithoutD}) is equivalent to
\begin{equation}\label{equ:DAEtransf}
    \begin{aligned}
    \dot{\Theta}_z(t) \!&=\! A_{11}\Theta_z(t)\! +\! A_{12}\Theta_y(t)\! +\! P_{1}\Lambda(t, t_0), \\ 
    \dot{\Theta}_y(t) \!&=\! A_{21}\Theta_z(t)\! +\! A_{22}\Theta_y(t)\! +\! P_{2}\Lambda(t, t_0)\! +\! b\Delta(t)\Lambda(t, t_0),  \\
    \bm{0}_{1 \times \nu} \!&= \Theta_y(t) + Q\Lambda(t, t_0), 
    \end{aligned}
\end{equation}
with
\begin{equation*}
    T A T^{-1} =\left[\begin{array}{ll}
        A_{11} & A_{12} \\
        A_{21} & A_{22}
    \end{array}\right], \quad 
    T P =  \left[\begin{array}{l}
        P_{1} \\
        P_{2}
    \end{array}\right].
\end{equation*}
Since $\Theta_y(t) = -Q\Lambda(t, t_0)$ for all $t \geq t_0$, the change of variable $\Theta \mapsto \bar{\Pi}_x := \Theta\Lambda^{-1} = \operatorname{col}\left(\bar{\Pi}_z, \bar{\Pi}_y\right)$ yields that~\eqref{equ:DAEtransf} is equivalent to
\begin{subequations}\label{equ:solvabCheckNFPi}
    \begin{align}
    &\quad\,\,
    \begin{aligned}\label{equ:solvabCheckNFIntDymPi}
        \bar{\Pi}_z(t) \!&=\! \Big( \!e^{A_{11} (t - t_0)} \bar{\Pi}_z(t_0)\Lambda(t_0, t_0) \\
     &\qquad  +\! \int_{t_0}^t\! e^{A_{11}(t-\tau)} G_1 \Lambda(\tau, t_0) d \tau \! \Big) \! \Lambda(t, t_0)^{-1}\!\!, 
    \end{aligned}
    \\ 
    &\begin{aligned} \label{equ:solvabCheckOutPi}
    \!Q\Lambda(t, t_0) \!&=\! Q\Lambda(t_0, t_0) \\
    &\quad + \!\!\int_{t_0}^t (A_{21}\bar{\Pi}_z(\tau) \!+\! G_2 \!+\! b\Delta(\tau))\Lambda(\tau, t_0) d\tau,
    \end{aligned}
    \end{align}
\end{subequations}
and $\bar{\Pi}_y(t) = Q$ for all $t \geq t_0$, with $\bar{\Pi}_z(t_0) \in \mathbb{R}^{(n-1) \times \nu}$, $G_1 = P_1 - A_{12}Q$, and $G_2 = P_{2} - A_{22}Q$. Since Lemma~\ref{lem:equivDAE} yields the equivalence between the quasi-regulator equations~\eqref{equ:reguEquation} and the DAEs~\eqref{equ:solutionDAE}, we can see that the solvability of~\eqref{equ:reguEquation} is equivalent to the existence of bounded piecewise continuous $\bar{\Pi}_z$ and $\Delta$ solving~\eqref{equ:solvabCheckNFPi} for all $t \geq t_0$. Note that under Assumptions~\ref{asmp:exoProp} and~\ref{asmp:pwDiff}, there always exists a bounded piecewise continuous $\Delta$ solving~(\ref{equ:solvabCheckOutPi}) if a bounded $\bar{\Pi}_z$ solving~(\ref{equ:solvabCheckNFIntDymPi}) exists for all $t \geq t_0$. For example, with a bounded $\bar{\Pi}_z$, there exists a bounded piecewise continuous $\Delta$ expressed by
\begin{equation*}
    \Delta(t) = b^{-1}\left(Q_\Lambda(t, t_0) - G_{2} - A_{21}\bar{\Pi}_z(t)\right),
\end{equation*}
which solves~(\ref{equ:solvabCheckOutPi}). Then, the existence of bounded piecewise continuous $\Pi_x$ and $\Delta$ solving the quasi-regulator equations~(\ref{equ:reguEquation}) comes down to the existence of a bounded solution $\bar{\Pi}_z$ in~(\ref{equ:solvabCheckNFIntDymPi}) for all $t \geq t_0$. When Assumptions~\ref{asmp:exoProp} and~\ref{asmp:pwDiff} hold, the vectorization of~(\ref{equ:solvabCheckNFIntDymPi}) yields that a bounded solution $\bar{\Pi}_z$ exists for any $\bar{\Pi}(t_0)$ and any $P$ and $Q$ if and only if the non-smooth non-resonance condition is satisfied. This concludes the proof.
\hfill $\blacksquare$
\end{pf}

Finally, we consider the regulation problem with $D \neq 0$, in which case the relative degree of system~(\ref{equ:system}) is zero and Assumptions~\ref{asmp:pwDiff} and~\ref{asmp:unitaryRD} can be omitted. Solvability of the quasi-regulator equations (\ref{equ:reguEquation}) can be studied in a similar way as in Theorem~\ref{thm:SolCondWithoutD}.

\begin{them}\label{thm:solCondWithD}
Consider the quasi-regulator equations~(\ref{equ:reguEquation}). Suppose $D \neq 0$ and Assumption~\ref{asmp:exoProp} holds. The quasi-regulator equations~(\ref{equ:reguEquation}) are solvable for any $P$ and $Q$ if and only if systems~(\ref{equ:system}) and~(\ref{equ:explicitGen}) are non-resonant.
\end{them}
\begin{pf}
By the substitutions $\Delta = -D^{-1}(C\Psi_x\Lambda^{-1} + Q)$ and $\hat{\Pi}_z = \Psi_x\Lambda^{-1}$, the quasi-regulator equations~(\ref{equ:reguEquation}) is solvable for any $P$ and $Q$ if and only if there exists a bounded piecewise continuous matrix-valued function $\hat{\Pi}_z(t) \in \mathbb{R}^{n \times \nu}$ of the form
\begin{equation*}
\begin{aligned}
    \hat{\Pi}_z(t) &=\! \Big( \!e^{A_{\pi}(t - t_0)} \hat{\Pi}_z(t_0)\Lambda(t_0, t_0) \\
    &\qquad \qquad+ \!\int_{t_0}^t e^{A_{\pi}(t-\tau)} P_{\pi} \Lambda(\tau, t_0) d \tau  \Big)  \Lambda(t, t_0)^{-1}\!, 
\end{aligned}
\end{equation*}
for all $t \geq t_0$ and any initial condition $\hat{\Pi}_z(t_0)$, where $A_{\pi} =  A - BCD^{-1}$ and $P_{\pi} = P - BQD^{-1}$. Under Assumption~\ref{asmp:exoProp}, this boundedness is satisfied for any $P$ and $Q$ if and only if systems~(\ref{equ:system}) and~(\ref{equ:explicitGen}) are non-resonant.
\hfill $\blacksquare$
\end{pf}




By Theorems~\ref{thm:SolCondWithoutD} and~\ref{thm:solCondWithD}, we have shown that the non-smooth non-resonance condition in Definition~\ref{def:nonResCond} is crucial for the quasi-regulator equations~\eqref{equ:reguEquation} to be solvable. It is important to point out that a sufficient condition for the non-resonance condition to hold is that the linear system~(\ref{equ:system}) is minimum phase.

\begin{propt}\label{prop:MiniPhaseNonRes}
    Suppose Assumption~\ref{asmp:exoProp} holds. Systems~(\ref{equ:system}) and~(\ref{equ:explicitGen}) are non-resonant if system~(\ref{equ:system}) is minimum phase.
\end{propt}
\begin{pf}
     When system~(\ref{equ:system}) is minimum phase and  Assumption~\ref{asmp:exoProp} holds, the first term on the right-hand side of~\eqref{equ:nonSmNRCOmega} is bounded and converges to zero when $t \to +\infty$ because $\Lambda(t_0, t_0)\Lambda(t, t_0)^{-1}$ is bounded for all times and $A_{z}$ is Hurwitz. Then the proof concludes if the integral term of~\eqref{equ:nonSmNRCOmega} is bounded for all times.     
     In fact, by Assumption~\ref{asmp:exoProp}, there exists a constant $h \in \mathbb{R}_{> 0}$ such that $\|  \Lambda(\tau, t_0)\Lambda(t, t_0)^{-1} \| \leq h$ for all $t \geq \tau \geq t_0$. Moreover, since the matrix $A_z$ is Hurwitz, there exist constants $\alpha, \beta \in \mathbb{R}_{> 0}$ such that $\|e^{A_z(t -\tau)}\| \leq \alpha e^{-\beta(t - \tau)}$. Consequently, the integral term of~\eqref{equ:nonSmNRCOmega} satisfies
     \begin{equation*}
    \label{eq-limPI_g}
    \begin{aligned}
       &\left\| \int_{t_0}^t \left(
    \Lambda(\tau, t_0)\Lambda(t, t_0)^{-1}\right)^{\top} \otimes e^{A_{z}(t - \tau)} d \tau \right\| \\
       &\qquad \leq \! \int_{t_0}^t \!\! \left\| \left(\! \Lambda(\tau, t_0)\Lambda(t, t_0)^{-1} \!\right)\!\!^{\top} \right\| \!\left\| e^{A_{z}(t - \tau)} \!\right\| \!d \tau \\
       &\qquad \leq \int_{t_0}^{t} h \cdot \alpha e^{-\beta(t-\tau)} d \tau 
       \leq \frac{\alpha h}{\beta} \left(\! 1 \!-\! e^{-\beta(t-t_0)} \!\right) \leq \frac{\alpha h}{\beta}.
    \end{aligned}
    \end{equation*}
    This concludes the proof.
    \hfill $\blacksquare$
\end{pf}

\begin{remark}
The non-smooth non-resonance condition is an extension of the non-resonance condition in the classical linear output regulation problem~\cite[Assumption 1.4]{ref:huang2004nonlinear} to the non-smooth case. To be more specific, if $\Lambda(t, t_0) = e^{S(t-t_0)}$ with $S \in \mathbb{R}^{\nu \times \nu}$, meaning that the signal generator~\eqref{equ:explicitGen} is equivalent to
\begin{equation*}
\dot{\omega}(t)=S \omega(t), \qquad \omega(t_0) = \omega_0,
\end{equation*}
the expression~(\ref{equ:solvabCheckNFIntDymPi}) becomes
\begin{equation}\label{equ:solvabCheckNFIntDymPiSylvester}
    \bm{0}_{(n-r)\times \nu} = A_{11}\bar{\Pi}_z - \bar{\Pi}_{z}S + G_1,
\end{equation}
and~(\ref{equ:nonSmNRCOmega}) becomes
\begin{equation*}
    \Omega(t)= e^{A_S(t - t_0)} \Omega(t_0) + \int_{t_0}^t e^{A_S(t - \tau)} d \tau.
\end{equation*}
where $A_S = I_\nu \otimes A_z - S^\top \otimes I_{n-r}$. Then the boundedness of $\Omega$, as required by the non-smooth non-resonance condition, implies (but is not implied by) the non-singularity of the matrix $A_S$. This is equivalent to $\sigma(A_{z}) \cap \sigma(S) = \varnothing$, which guarantees the solvability of~\eqref{equ:solvabCheckNFIntDymPiSylvester} as the eigenvalues of $A_{z}$ coincide with the transmission zeros of system~(\ref{equ:system}), \textit{i.e.}, $\sigma(A_{z}) = \sigma(A_{11})$. Thus, we can see that as a result of extending the Sylvester equation~\eqref{equ:solvabCheckNFIntDymPiSylvester} to the non-smooth non-periodic case, the derived non-resonance condition in Definition~\ref{def:nonResCond} becomes stronger than the classical non-resonance condition, see~\cite{ref:bhatia1997and} for more discussion of the integral-from solution of the Sylvester equations.
\end{remark}


\begin{remark}
    While the established solvability results guarantee the existence of solutions to the quasi-regulator equations~\eqref{equ:reguEquation}, these solutions, $\Pi_x$ and $\Delta$, are not necessarily unique in the non-smooth, non-periodic case. To see this, consider the scenario where system~\eqref{equ:system} is minimum phase, which, by Proposition~\ref{prop:MiniPhaseNonRes}, guarantees the satisfaction of the non-smooth non-resonance condition. As a result, $A_{11}$ in~\eqref{equ:solvabCheckNFIntDymPi} is Hurwitz, and~\cite[Lemma~1]{ref:Niu2025arXiv} implies that the solution $\bar{\Pi}_z$ starting with any initial condition $\bar{\Pi}_z(t_0)  \in \mathbb{R}^{(n-r) \times \nu}$ is bounded, indicating the nonuniqueness of the bounded piecewise continuous solutions $\bar{\Pi}_z$ and $\Delta$ solving~\eqref{equ:solvabCheckNFPi}. By the equivalence between~\eqref{equ:reguEquation} and~\eqref{equ:solvabCheckNFPi} discussed in the proof of Theorem~\ref{thm:SolCondWithoutD}, this non-uniqueness directly extends to the quasi-regulator equations~\eqref{equ:reguEquation}. A similar conclusion also holds for the case $D \neq 0$ addressed in Theorem~\ref{thm:solCondWithD}.
\end{remark}


\section{Conclusion}\label{sec:concl}

In this paper, we have addressed the solvability of the quasi-regulator equations arising from the NNOR problem. To this end, we have first equivalently reformulated the quasi-regulator equations as a group of DAEs. By separately studying the DAEs with and without the feedforward term, we have elucidated the significance of the relative degree in determining solvability, which is also related to the smoothness properties of the exogenous signal. Additionally, under specific relative degree requirements, we have proposed a non-smooth, non-resonance condition that provides a necessary and sufficient characterization of the solvability of the quasi-regulator equations.



\bibliography{ifacconf}             
                                                   







\end{document}